\documentclass{elsarticle}
\usepackage{amssymb,amsmath,amsthm}
\usepackage{stmaryrd}
\usepackage{delarray}
\usepackage{enumitem}
\usepackage{tikz}
\usepackage{hyperref}
\graphicspath{{figures/}}
\usepackage{ifthen}
\usepackage{pifont}
\usepackage{bitset}
\usepackage{pgf,pgffor}
\usepackage{caption}
\usepackage{subcaption}
\usepackage[lined,boxed,commentsnumbered,ruled]{algorithm2e}
\usepackage{xparse}
\usepackage{placeins}

\SetKwFor{ParFor}{for}{pardo}{end}

\graphicspath{{figures/}}

\newtheorem{theorem}{Theorem}
\newtheorem{lemma}{Lemma}
\newtheorem{definition}{Definition}

\renewcommand{\int}[1]{[#1]}
\newcommand{\intz}[1]{\llbracket#1\rrbracket}

\NewDocumentCommand{\Pred}{o}{%
  \IfNoValueTF{#1}
    {\textsf{Prediction} }%
    {#1\textsf{Prediction}}%
}%

\newcommand{\N}{\mathbb{N}}
\newcommand{\Z}{\mathbb{Z}}

\newcommand{\PP}{{$\mathsf{P}$}}
\newcommand{\PC}{\PP-Complete}
\newcommand{\NC}{{$\mathsf{NC}$}}

\newcommand{\neigh}{\mathcal{N}}

\begin{document}

\begin{frontmatter}
    
\title{Complexity of the Freezing Majority Rule with L-shaped Neighborhoods}
\author[1]{Pablo Concha-Vega\corref{cor1}}\ead{pablo.concha-vega@lis-lab.fr}
\author[2,3]{Eric Goles}\ead{eric.chacc@uai.cl}
\author[2,3]{Pedro Montealegre}\ead{p.montealegre@uai.cl}
\author[1,4]{K{\'e}vin Perrot}\ead{kevin.perrot@lis-lab.fr}
\affiliation[1]{organization={Aix-Marseille Universit{\'e} Toulon, LIS, CNRS UMR 7020}, city={Marseille}, country={France}}
\affiliation[2]{organization={Facultad de Ciencias y Tecnolog{\'i}a, Universidad Adolfo Ib{\'a}{\~n}ez}, city={Santiago}, country={Chile}}
\affiliation[3]{organization={Millennium Nucleus for Social Data Science (SODAS)}, city={Santiago}, country={Chile}}
\affiliation[4]{organization={Université publique}, country={France}}
\cortext[cor1]{Corresponding author}

\begin{abstract}
    In this article we investigate the computational complexity of predicting
    two dimensional freezing majority cellular automata with states \(\{-1,+1\}\),
    where the local interactions are based on an L-shaped neighborhood structure.
    In these automata, once a cell reaches state \(+1\),
    it remains fixed in that state forever, while cells in state \(-1\) update to the most
    represented state among their neighborhoods. We consider L-shaped neighborhoods, which
    mean that the vicinity of a given cell \(c\) consists in a subset of cells in the north
    and east of \(c\).

    We focus on the prediction problem, a decision problem that involves determining
    the state of a given cell after a given number of time-steps. 
    We prove that when restricted to the simplest L-shaped neighborhood,
    consisting of the central cell and its nearest north and east neighbors, the prediction problem
    belongs to {\NC}, meaning it can be solved efficiently in parallel.
    We generalize this result for any L-shaped neighborhood of size two.
    On the other hand, for other L-shaped neighborhoods, the problem becomes {\PC},
    indicating that the problem might be inherently sequential. 
\end{abstract}


\end{frontmatter}


\section{Introduction}

We study the computational complexity of a class of two-dimensional majority
cellular automata defined on a lattice where each node takes a binary state
from the set \(\{-1, +1\}\). The system evolves in discrete time steps, and each
node updates its state in parallel by adopting the majority value within a
prescribed neighborhood. In the event of a tie, the node retains its previous
state. Crucially, once a node switches to state \(+1\), it becomes frozen and
cannot revert. This freezing behavior captures irreversible dynamics observed
in various contexts, such as opinion formation in social networks
(voting models) and distributed consensus mechanisms
(see~\cite{dwork1988consensus,galam2007role,galam2008sociophysics,wang2014freezing}
for example).

The computational question we address is: how hard is it to predict the behavior
of such automata after a given number of steps? More precisely, given an initial
configuration and a specific node, can we efficiently determine its state after $t$
rounds of updates? This question lies at the intersection of dynamical systems and
computational complexity theory. It is known that for majority rules on
three-dimensional lattices, the prediction problem is {\PC}, implying that no
efficient parallel algorithms are likely to exist. In contrast, for one-dimensional
systems, the problem belongs to the class {\NC}, and thus admits polylogarithmic-time
parallel algorithms~\cite{moore1997majority}.

In the two-dimensional case with von Neumann neighborhoods (i.e., each node
considers its four immediate neighbors), the complexity of the
prediction problem remains unresolved. Nevertheless, {\PC}ness has been shown
for a specific two-dimensional variant involving heterogeneous majority
thresholds~\cite{concha2022complexity}. However, progress has been made in
restricted versions of the original model. For instance, it was shown
in~\cite{goles2013complexity} that for \emph{freezing} majority automata---
where the state \(+1\) remains stable--- an efficient parallel algorithm exists,
placing the problem in {\NC}.

This article contributes to this line of research by exploring the complexity of
majority dynamics under \emph{L-shaped} neighborhoods, which consist of the central
cell and a subset of its adjacent neighbors forming an ``L'' (e.g., the north and
east neighbors). We prove a dichotomy: when the neighborhood is the minimal L-shape
(central cell plus top and right neighbors), the freezing majority problem is in {\NC}.
However, when the L-shape includes more neighbors (i.e., larger L patterns),
the prediction problem becomes {\PC}.

Interestingly, the minimal L-shaped neighborhood, known as the Toom neighborhood,
also appears in the study of \emph{eraser cellular automata},
where the goal is to eliminate finite ``islands''
of active cells (state \(+1\)) over time~\cite{toom1980stable,gacs2022stable}.
In this context, the majority rule helps characterize local update rules that
ensure convergence to the all-zero configuration.
Related work in~\cite{gacs1978one} introduced a non-connected majority rule that alternates
the neighborhood depending on the current state of the node. This rule successfully
solves the majority problem in most cases but fails when the initial configuration
has near-zero magnetization (i.e., both opinions are equally represented).

\paragraph*{Our contribution.}

We first prove that the prediction problem lies in {\NC} for the Toom neighborhood
(Theorem~\ref{thm:toom_nc}), and then generalize this result to all L-shaped of
size two (Theorem~\ref{thm:size1_nc}). Next, we examine larger neighborhoods,
starting with those consisting of contiguous cells in both directions
(at least two on each side), for which the problem is {\PC}
(Theorem~\ref{thm:connected}). We also establish {\PP}-completeness for
neighborhoods of equally spaced cells (Theorem~\ref{thm:periodic});
in this case, the spacing between cells on the upper side and the right side
can differ. Finally, we show that neighborhood with two cells on each side also
lead to a {\PC} prediction problem, regardless of the distance between the cells.

\section{Preliminaries}

We use the notations $\int{n} = \{1,\dots,n\}$ and $\intz{n}=\{0,\dots,n-1\}$.
Let $\N$ be the set of positive natural numbers \(\{1,2,\dots\}\).
We define $\N_0=\N \cup \{0\}$.
We denote $V_n = \int{n}^2 \subset \N_0^2$, the grid of $n \times n$ nodes with periodic boundary conditions.

\paragraph*{Cellular Automata.}
Cellular automata are discrete dynamical systems defined on a regular grid of cells,
where each cell changes its state following a local rule which depends on the state
of the cell and the states of its neighbors.
In this work, we consider a two-dimensional grid over a torus
(periodic boundary conditions), where the cells are arranged in a rectangle
of square cells. In this model, cells can only have a finite number of states.
The states are \(-1\) and \(+1\).

A \textit{configuration} is a function \(x\) that assigns values
in \(\{-1, +1\}\) to a region of the grid \(V_n\). The value of the
cell \(u\) in the configuration \(x\) is denoted as \(x_u\). For a given
cell \(u\) we will refer to \(\neigh(u)\) as its neighborhood, which is by definition
a finite set of cells. In addition, we call
\(x_{\neigh(u)}\) the restriction of \(x\) to the neighborhood of \(u\).
For a cellular automaton, the size of the neighborhood of a cell is uniform,
\textit{i.e.}, \(|\neigh(u)|\) is the same for every cell \(u\).
Moreover, we have that \(\neigh(u) = \neigh(0) + u \pmod{n}\).
Neighborhoods are taken modulo \(n\), due to the periodic setting.

Formally, a two-dimensional \textit{cellular automaton} (CA) with states \(Q\) and local function \(f: Q^{|\neigh(u)|} \to Q\) is a map \(F: Q^{n^2} \to Q^{n^2}\) such that \(F(x)_u = f(x_{\neigh(u)})\) We call \(F\) the \textit{global function} or the \textit{global rule} of the CA. The dynamic is defined by assigning to the configuration \(x\) a new state given by the synchronous update of the local function on \(x\).

\paragraph*{Freezing majority.}
In this work, we study the freezing majority rule. In the majority cellular automata, each cell changes its state to the most represented one in its neighborhood. On the other hand, the freezing property means that a cell in state \(+1\) remains fixed in every future step.

Let us define the Freezing Majority Cellular Automata (FMCA) by the local function:

\begin{equation*}
    F(x)_u = 
    \begin{cases}
        +1 & \text{ if } x_u = +1 \text{ or } \sum_{v \in \neigh(u)} x_v > 0,\\
        -1 & \text{ otherwise.}
    \end{cases}
\end{equation*}

\paragraph*{L-shaped neighborhoods.}
We define an L-shaped neighborhood (Figure~\ref{fig:neigh_examples})
with two finite sets \(S_N,S_E \subsetneq \N \) as: 

\[
    \neigh(i,j) = \{(i,j+k), k \in S_N\} \cup \{(i+k,j), k \in S_E\}.
\]

In this fashion, we define the {\it L-shaped Freezing Majority Cellular Automata}
(LFMCA) as FMCAs with L-shaped neighborhoods.


\begin{figure}[t]
    \centering
    \includegraphics[scale = .5]{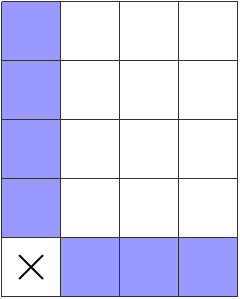}
    \qquad
    \includegraphics[scale = .5]{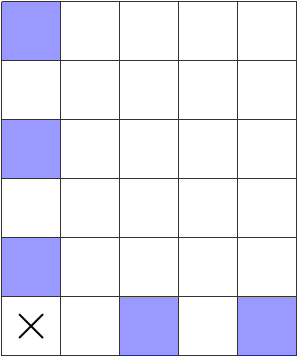}
    \qquad
    \includegraphics[scale = .5]{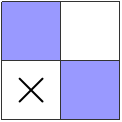}
    \qquad
    \includegraphics[scale = .5]{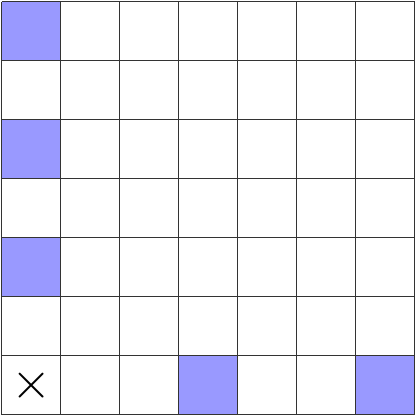}
    \caption{Examples of L-shaped neighborhoods, where the cells marked with an $\times$
    represent the central cells. From left to right, they are defined by the sets:
    $S_N = \{1,2,3,4\}$ and $S_E = \{1,2,3\}$;
    $S_N = \{1,3,5\}$ and $S_E = \{2,4\}$;
    $S_N = S_E = \{1\}$; and
    $S_N = \{2,4,6\}$ and $S_E = \{3,6\}$.
    }
    \label{fig:neigh_examples}
\end{figure}

\paragraph*{The prediction problem.}
In this paper, we focus on the prediction problem for LFMCA,
a decision problem that asks
whether the state of a given cell will change after \(t\)
time steps. We provide a precise definition below:
\newline
\newline
\noindent\fbox{%
    \parbox{.99\textwidth}{%
        \Pred[\(S_N\)-\(S_E\)-]
        \begin{itemize}
            \item {\bf Input}: \begin{itemize}
                \item an initial configuration \(x \in \{-1,1\}^{V_n}\),
                \item a time \(t \in \N\), and 
                \item a cell \((i,j) \in V_n\).
            \end{itemize}
            \item {\bf Question}: \(F^t(x)_{(i,j)} \not = x_{(i,j)}?\)
        \end{itemize}
    }
}
\vspace{0.1in} 

\paragraph*{Complexity classes.}
The computational complexity of a computational task can be defined
as the amount of resources, like time or space, needed to solve it.
One fundamental set of computational decision problems is the class
{\PP}, which is the class of problems solvable in polynomial time on
a deterministic Turing machine. The class {\PP} is informally known
as the class of problems that admit an efficient algorithm. 

Evidently, \Pred can be solved by simply simulating the dynamics of
the cellular automaton for the given number of time-steps. This is
called the \emph{trivial algorithm}. Furthermore, every initial
configuration can only reach a fixed point (by the freezing property),
and at each time step before reaching the attractor at least one cell
gets frozen on state \(+1\). Hence, the dynamics of any initial
configuration reaches an attractor in at most \(n^2\) time-steps.
Therefore, if we aim to solve the prediction problem more efficiently
than the trivial algorithm, we would have to classify it in a subclass
of {\PP}.

We consider the class {\NC}, a subclass of {\PP}, 
which contains the problems solvable in poly-logarithmic time by a PRAM,
with a polynomial number of processors.
{\NC} is considered as the class of problems which admit a fast parallel algorithm. 
It is a well known conjecture that {\NC} \(\neq\) {\PP} and, if so,
there exist ``inherently sequential''
problems that belong to {\PP} and do not belong to {\NC}.
The most likely to be inherently sequential are {\PC} problems,
to which any other problem in {\PP} can be reduced (by an
{\NC}-reduction or a logarithmic space reduction).
If any of these problems has a fast parallel algorithm,
then {\PP} \(=\) {\NC} \cite{JaJa:1992:IPA:133889,greenlaw1995limits}.

A  well known {\PC} problem is the Circuit Value Problem (CVP),
which consists in, given a Boolean circuit and a truth-value of its inputs,
to determine the value of the output.
Further, the Monotone Circuit Value Problem (MCVP), \textit{i.e.},
considering only AND and OR gates, is also {\PC}.
Intuitively, circuits are hard to compute efficiently in parallel
because in order to compute a given layer of gates, it is necessary
to know the state of the previous ones.
For a deeper understanding of circuit complexity, see 
\cite{wegener1987complexity,vollmer1999introduction,epstein2000computability}.

\section{The Complexity of the FMCA with Toom Neighborhood}

In this section, we prove that \Pred[\(\{1\}\)-\(\{1\}\)-], \emph{i.e.},
the prediction problem for the LFMCA with the Toom neighborhood, belongs to {\NC}.
We then extend this result to any L-shaped neighborhood of size \(2\).

Given a configuration \( x \in \{-1,+1\}^{V_n} \), we define the digraph
 \( G_x = (V_x, E_x) \) as follows:
\begin{itemize}
  \item The set of vertices \( V_x \subseteq V_n \) consists of all cells in state \(-1\), \emph{i.e.},

  \[V_x = \{ (i,j) \in V_n \mid x_{i,j} = -1 \}.\]
  \item There is a directed edge from \((i,j)\) to \((i',j')\) in \( E_x \) if and only if:
  
  \[  (i,j), (i',j') \in V_x \quad \text{and} \quad (i',j') \in \mathcal{N}(i,j),\]

  where \( \mathcal{N}(i,j) = \{(i+1,j), (i,j+1)\} \) corresponds to the Toom neighborhood.
\end{itemize}

The following Lemma captures
the key idea that enables an efficient parallel solution.

\begin{lemma}\label{lemma:toom_fixed}
    Given the LFMCA defined by \(S_N = S_E = \{1\}\),
    an initial configuration \(x \in \{-1,+1\}^{V_n}\),
    and the associated digraph \(G_x = (V_x, E_x)\).
    A cell \(v \in V_x\) will remain fixed in state \(-1\)
    if and only if it is in a cycle or has a path to a cycle
    of \(G_x\).
\end{lemma}

\begin{proof}
    Let us begin by proving that if \(v\) is in a cycle of \(G_x\),
    then it remains fixed in state \(-1\). 
    First, if \(v\) is in a cycle,
    it means it has at least one neighbor in state~\(-1\) at \(t = 0\).
    We can assume that the other neighbor is in state~\(+1\),
    since this implies that cell \(v\) will change its state if and only if
    the neighbor in state \(-1\) also changes. 
    We can also extend this assumption for every node in the cycle.
    Since every node in the cycle is in state \(-1\),
    and that they depend of each other to change their states,
    none of them will change.
    Something similar occurs when \(v\) is in a path to a cycle:
    it depends on the path nodes which in turn depend of the cycle,
    then it remain fixed in state \(-1\).

    On the other hand, let us prove that if there does not exist a time-step \(t\)
    such that \(F^t(x)_v = +1\), 
    it implies that \(v\) is in a cycle or connected to a cycle of \(G_x\).
    If \(v\) is always in state \(-1\), 
    it means that at least one if its neighbors is also always in state \(-1\),
    which in turn always has a neighbor in state \(-1\), and so on.
    If we continue with this reasoning,
    at some point we will reach the ``end'' of the grid
    (remember we are considering periodical border conditions),
    and moreover, we will reach an already visited node. 
    This node can be \(v\), which means it is in a cycle.
    If it is not \(v\), it means \(v\) is in a path to the cycle.
\end{proof}

Later we will need to compute the prefix sum of an array.
Given a group \((X, +)\), the prefix sum of a sequence of elements
\(x_0,x_1,\dots, x_n \in X\) is another sequence \(y_0,y_1, \dots, y_n\),
where \(y_k = x_0 + x_1 + \dots + x_n\), with \(0 \leq k \leq n\).
J\'{a}J\'{a} \cite{JaJa:1992:IPA:133889} provides us with an efficient {\NC} 
algorithm for computing the prefix sum. It also supplies us with an efficient
algorithm for matrix multiplication. We are now ready to show the main result
of this section. 

\begin{theorem}\label{thm:toom_nc}
    \Pred[\(\{1\}\)-\(\{1\}\)-]~is in {\NC}.
\end{theorem}

\begin{proof}
    To prove it, we need to compute the graph \(G_x\) from 
    Lemma~\ref{lemma:toom_fixed} and then check whether the given node \(v\) is in
    a cycle, connected to a cycle or neither.
    Keep in mind that everything have to be done in poly-logarithmic time
    with a polynomial number of processors.

    For calculating \(G_x\) (Algorithm \ref{alg:compt_oraG}),
    we assign a processor to every node of the graph \(G\).
    If the node is in state \(+1\), the processor does nothing, otherwise,
    it will store in memory a pointer to every neighbor in state \(-1\).
    This is done in constant time using \(n^2\) processors.

    \begin{algorithm}[t]
        \caption{Computing \(G_x\).}\label{alg:compt_oraG}
        \KwIn{An initial configuration \(x\) (size \(n \times n\)).}
        \KwOut{Adjacency matrix \(A\) of graph \(G_x\) (size \(n^2 \times n^2\)).}
        \ParFor{\(1 \leq i,j \leq n\)}{
            \If{\(x_{i,j} = -1\)}{
                \If{\(x_{i+1,j} = -1\)}{
                    \(A[(i,j),(i+1,j)] = 1\)
                }
                \If{\(x_{i,j+1} = -1\)}{
                    \(A[(i,j),(i,j+1)] = 1\)
                }
            }
        }
    \end{algorithm}

    For checking whether node \(v\) is in a cycle or in a path to a cycle,
    we first compute the powers of matrix \(A\) of \(G_x\), \textit{i.e.}, \(A, A^2, \dots, A^{n^2}\),
    which can be done by mixing the prefix sum and matrix multiplication techniques
    provided by~\cite{JaJa:1992:IPA:133889}.


    To check if \(v\) is in a cycle, it is enough to read the entry \((i,j),(i,j)\)
    of every power of \(A\). This can be done in constant time (Algorithm \ref{alg:in_cycle}).

    \begin{algorithm}[t]
        \caption{Checking if \(v = (i,j)\) is in a cycle of \(G_x\).}\label{alg:in_cycle}
        \KwIn{Powers of \(A\): \(P = A, A^2, \dots, A^{n^2}\)}
        \KwOut{\(\textsf{answer = 1}\) if \(v\) is in a cycle of \(G_x\), \(\textsf{answer = 0}\) otherwise}
        \(\textsf{answer = 0}\)\\
        \ParFor{\(1 \leq k \leq n^2\)}{
            \If{\(P[k][(i,j),(i,j)] = 1\)}{
                \(\textsf{answer = 1}\)
            }
        }
    \end{algorithm}

    Finally, to decide whether \(v\) is in a path to a cycle, we use the following technique.
    Let us take the adjacency matrix \(A\) of graph \(G_x\), 
    and add a new node named \(u\), such that every node in a cycle points to it.
    Let us call this new matrix \(B\).
    We can compute \(B\) by parallel executing the Algorithm \ref{alg:in_cycle} on every node.
    Next, we calculate the first \(n^2+1\) powers of \(B\).
    Finally, it is enough to check if there exists a \(k \leq n^2+1\) such that
    \(B^k[v,u] = 1\). The latter can be done similarly to Algorithm~\ref{alg:in_cycle}.
\end{proof}

In the following results, we prove that Theorem~\ref{thm:toom_nc}
can be extended to arbitrary sets $S_N$ and $S_E$ of size $1$.

First, let us define the concept of \textit{independent subgrid}.

\begin{definition}[Independent subgrid]
    Given a CA over \(G = (V_n,E)\), where the directed edges \(E\)
    are determined by a neighborhood function \(\neigh\). \(G'\) is
    an independent subgrid of \(G\) if and only if it is a
    maximal size induced subgraph of \(G\).
\end{definition}

Note that a CA with more than one independent subgrid will have independent
subdynamics, therefore solving \Pred would depend only on one of these subdynamics.
We state the following Lemma in a general way as it is also used later.
Remark that the existence of subgrids also depends on the size of the finite square grid $V_n$,
because the latter has periodic boundary conditions
(for example with $S_N=S_E=\{2\}$ and $n=7$, the neighborhood relationship on $V_7$ is isomorphism to Toom's neighborhood on $V_7$).

\begin{lemma}\label{lemma:neigh_equiv}
    Given an LFMCA defined with sets \(S_E\) and \(S_N\) such that
    \begin{itemize}
        \item \((\exists p \in \N)(\forall a \in S_E): a \equiv 0 \mod p\), and 
        \item \((\exists p' \in \N)(\forall b \in S_N): b \equiv 0 \mod p'\)
    \end{itemize}
    and a size $n\in\N$,
    then for every \((i,j) \in \intz{\gcd(p,n)} \times \intz{\gcd(p',n)}\)
    we have independent subgrids composed by the vertices (coordinates are considered modulo $n$)
    \[
      V_{(i,j)} = \{(i+qp , j+q'p') \in V_n \mid q,q' \in \N_0\}.
    \]
\end{lemma}

\begin{proof}
    We prove that the \(V_{(i,j)}\) for \((i,j) \in \intz{\gcd(p,n)} \times \intz{\gcd(p',n)}\)
    form a partition of the grid $V_n$.
    First, they do not intersect, because if $i\neq i'$ for some $i>i'$
    (the case $j\neq j'$ for the second coordinate is analogous)
    then we have $i+qp \not\equiv i'+q'p \mod n$ for all $q,q'\in\N_0$.
    Indeed, otherwise it should hold that $i-i'\equiv (q'-q)p \mod n$, but this is impossible because $1\leq i-i'<\gcd(p,n)$.

    Second, their union is $V_n$, because for any $0\leq i'<n$ we have $i'=i+qp\mod n$
    for some $q\in\N_0$ and $i\in\intz{\gcd(p,n)}$
    (similarly for any second coordinate $0\leq j'<n$).
    Indeed, by Bezout's identity there are $x,y\in\Z$ such that $px+ny=\gcd(p,n)$,
    hence $px\equiv\gcd(p,n)\mod n$ (Property~$\star$),
    and since $p(x+n)\equiv\gcd(p,n)\mod n$ we can assume $x\in\N_0$ with Property~$\star$.
    Let $r=\lceil\frac{i'}{\gcd(p,n)}\rceil$ and $s=i'-r\gcd(p,n)$,
    we have $i'-rxp \equiv s\mod n$,
    therefore we conclude that there are $q=rx\in\N_0$ and $i=s\in\intz{\gcd(p,n)}$
    verifying the claim.

    Moreover, by construction all the neighborhood relationships are contained in the parts
    (\emph{i.e.}~if cell $(i,j)\in V_{(s,t)}$ is neighbor of cell $(i',j')\in V_{(s',t')}$ then $(s,t)=(s',t')$).
\end{proof}

\begin{theorem}\label{thm:size1_nc}
    The \Pred[\(S_N\)-\(S_E\)-] problem is in {\NC} for all sets
    $S_N$ and $S_E$ such that $|S_N| = |S_E| = 1$.
\end{theorem}

\begin{proof}





    The proof follows directly from Theorem~\ref{thm:toom_nc} and
    Lemma~\ref{lemma:neigh_equiv}.

    In the partitions of this LFMCA for any size $n\in\N$, each cell has one neighbor to the top and
    another to the right, similar to the Toom neighborhood.
    The latter implies that, after identifying in which partition the objective cell is,
    the \Pred problem can be solved as we proved in Theorem~\ref{thm:toom_nc}.
    
    Let $S_N=\{k_N\}$ and $S_E=\{k_E\}$.
    In order to identify in which partition a given cell $(i,j)\in \Z^2$ belongs,
    it is enough to calculate the values of the expressions
    \[
      s = i \mod gcd(k_E,n)
      \quad\text{ and }\quad
      t = j \mod gcd(k_N,n)
    \]
    then cell $(i,j) \in V_{(s,t)}$.
\end{proof}

\section{Hardness of the General L-shaped Neighborhoods} 

In this Section we study L-shaped neighborhoods of greater size.
We do not only study the obvious generalization of the Toom neighborhood
consisting in two lines of consecutive cells pointing north and east,
but also two other families within our definition of L-shaped neighborhood,
all of them having a {\PC} \Pred problem. These two families are
1) the L-shaped neighborhoods  where the elements of \(S_N\) and \(S_E\)
are equally spaced and 2) the L-shaped neighborhoods defined by \(S_N\)
and \(S_E\) such that \(|S_N| = |S_E| = 2\).

\subsection{Contiguous and Periodic L-Shaped Neighborhoods}

Our next result focuses on the natural generalization of the Toom neighborhood,
\emph{i.e.}, the neighborhoods defined with arbitrary numbers of subsequent
cells in both directions.

\begin{theorem}\label{thm:connected}
    The \Pred[\(S_N\)-\(S_E\)-] problem is {\PC} for all sets
    \(S_E = \int{k_E}\) and \(S_N = \int{k_N}\) such that \(k_E,k_N \ge 2\).
\end{theorem}

\begin{proof}
    To establish the result, we reduce from the MCVP. The reduction encodes the input
    and the components of a Boolean circuit into an initial configuration of the LFMCA,
    which depends on the sets \(S_E\) and \(S_N\).
    Without loss of generality, we assume \(k_E \ge k_N\).
    
    Following the approach described in \cite{goles2018complexity}, it is enough to
    construct a fixed set of \textit{gadgets}, namely \textit{wires}, conjunction gates,
    disjunction gates, a crossover and signal duplicators. Each gadget have to be of
    constant size and designed with exactly two inputs and two outputs, with inputs
    located on the east and north side and outputs on the west and south side.
    Moreover, whenever two gadgets are placed adjacent to each other, an output port
    of one must align with an input port of the other. We adopt this strategy and
    construct gadgets that satisfy all these requirements.

    First, we explain how to simulate \textit{wires}, \textit{i.e.}, gadgets that
    allow us to propagate signals through the grid. Since \(k_E\) and \(k_N\) are not
    necessarily equal, the vertical and horizontal wires may differ in shape.
    Hence, we will see both cases separately.
    As shown in Figure \ref{fig:vert_connected_wires}, vertical wires are constructed
    using \(a = \left\lfloor \frac{k_E + k_N}{2} \right\rfloor\) columns of cells in
    state \(+1\), the signal passes through $c+1$ columns located \(b = k_E - a\) columns
    to the left, where

    \[
        c = 
        \begin{cases}
            b+1 & \text{ if } k_E + k_N \text{ is even,}\\
            b & \text{ otherwise.}
        \end{cases}
    \]

    To encode a TRUE signal, the first row of the signal columns must be in state \(+1\).

    \begin{figure}[t]
        \centerline{
        \includegraphics[scale = .3]{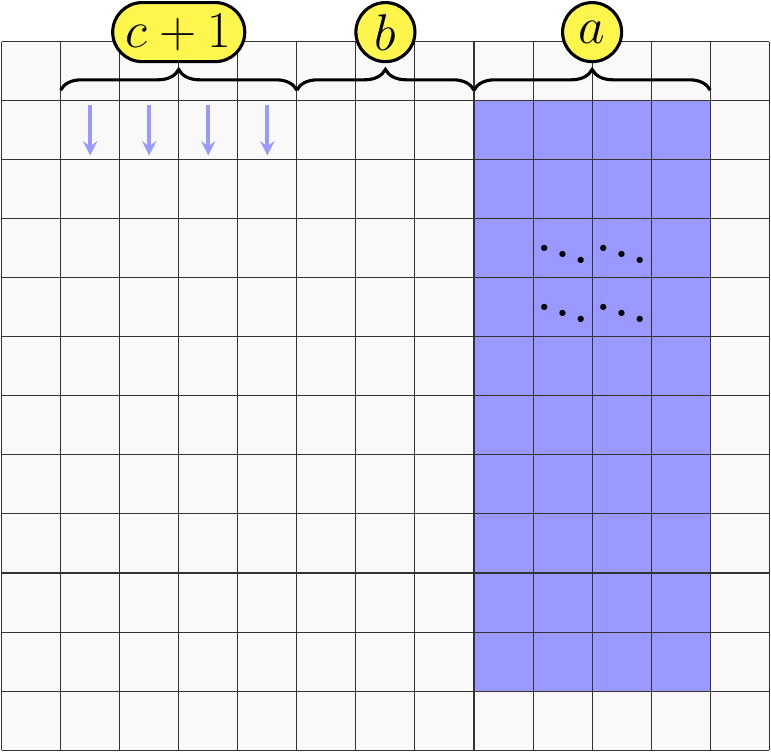}\qquad
        \includegraphics[scale = .3]{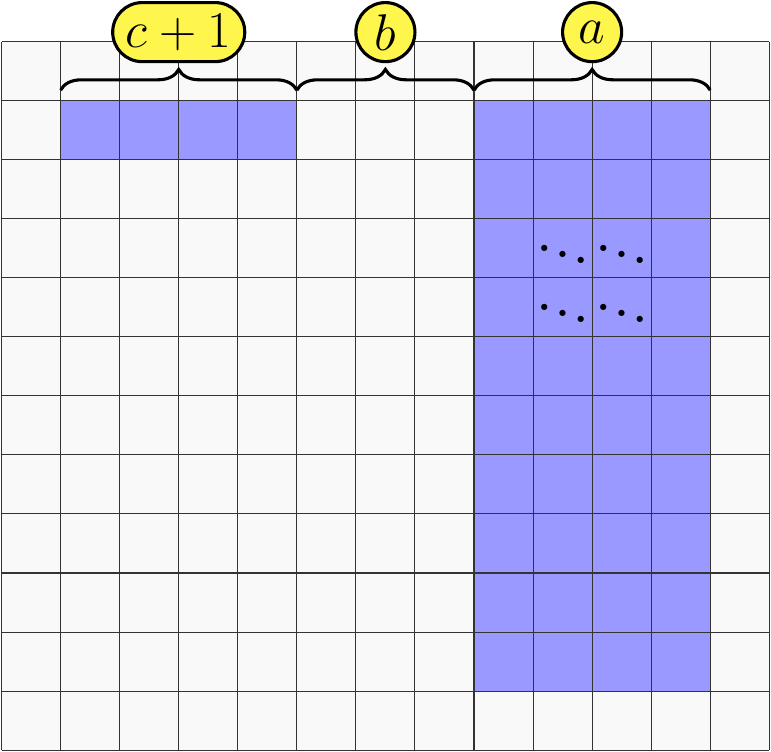}
        }
        \vspace{1em}
        \centerline{
        \includegraphics[scale = .3]{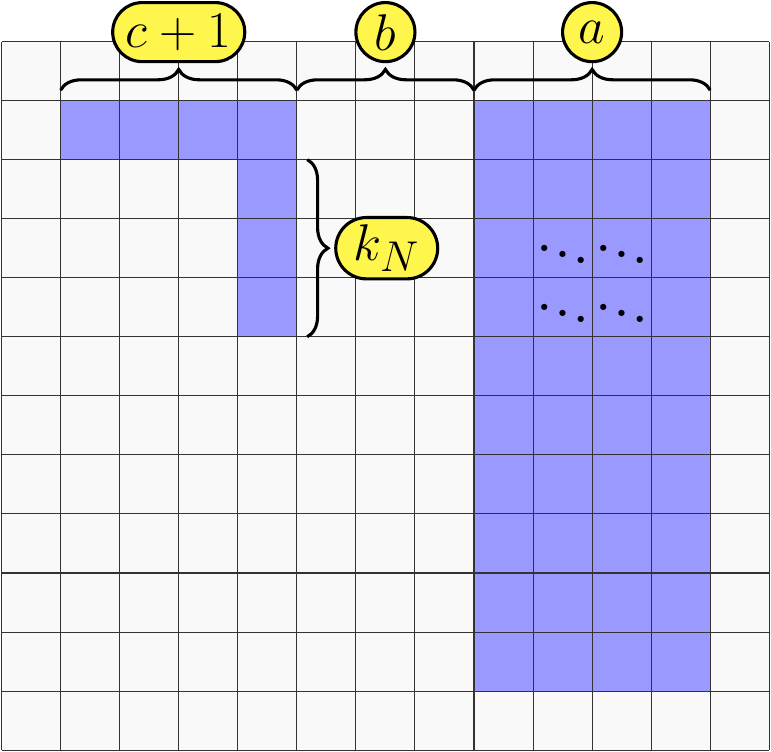}\qquad
        \includegraphics[scale = .3]{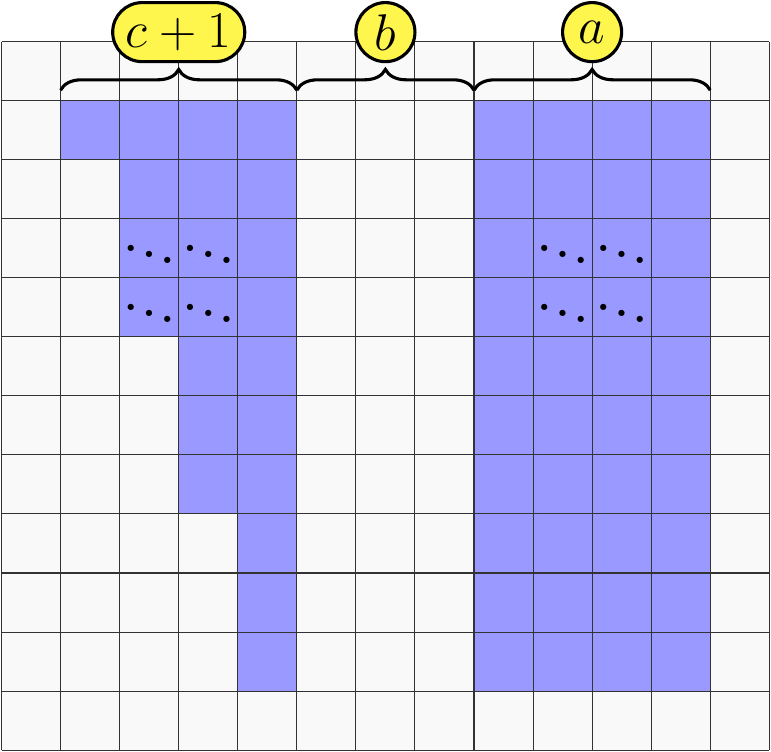}
        }
        \caption{Vertical wires. \textbf{Top left:} Wire initialized with a FALSE signal.
        The arrow indicates the cell that must be in state \(+1\) to \textit{activate} the signal.
        \textbf{Top right:} Wire initialized with a TRUE signal. \textbf{Bottom left:} Wire initialized with a TRUE signal after one step. \textbf{Bottom right:} Wire with a TRUE signal after several steps.}
        \label{fig:vert_connected_wires}
    \end{figure}

    On the other hand, horizontal wires (see Figure \ref{fig:hor_connected_wires}) 
    are composed of \(k_N\) rows in state \(+1\), 
    with the signal located in the row immediately below.
    Horizontal signals have width $2$, which helps with the construction of
    the other gadgets.
    A TRUE signal is initialized by setting the first \(c\) cells of both signal
    rows to state \(+1\).

    Note that the horizontal wires cannot occupy more than \(a\) columns, 
    otherwise undesired horizontal spreading occurs across all the rows they use.
    However, they can be placed up to \(k_E + k_N - a - 2\) 
    columns apart and still function properly as wires without this undesired
    row-wise contamination.

    \begin{figure}[t]
        \centerline{
        \includegraphics[scale = .3]{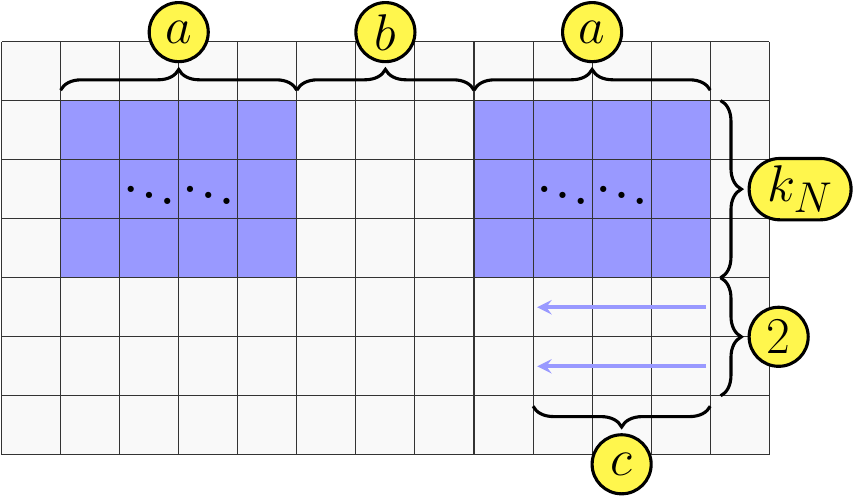}\qquad
        \includegraphics[scale = .3]{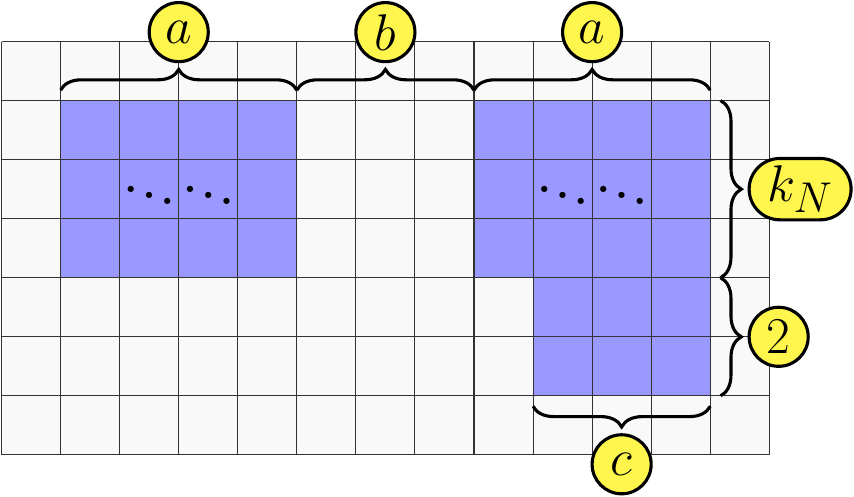}
        }
        \vspace{1em}
        \centerline{
        \includegraphics[scale = .3]{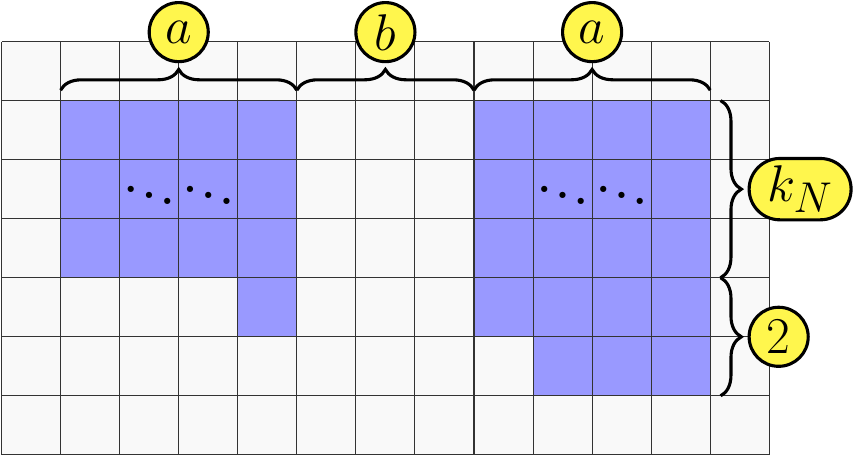}\qquad
        \includegraphics[scale = .3]{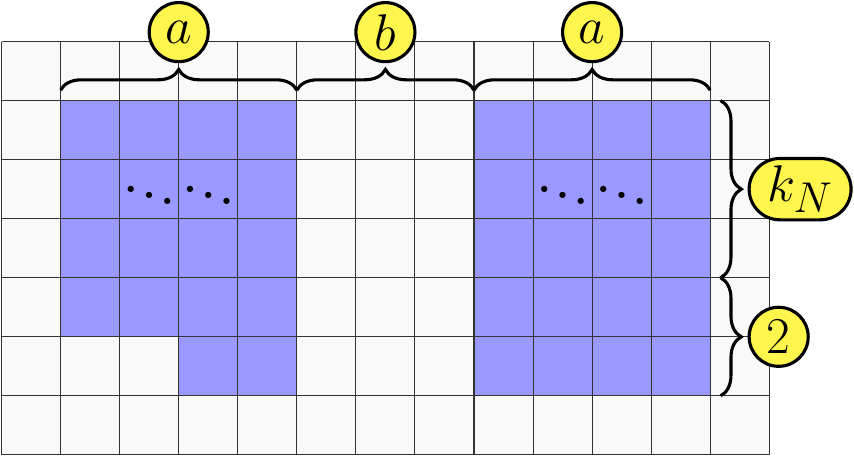}
        }
        \caption{Horizontal wires.
        \textbf{Top left:} Wire initialized with a FALSE signal,
        the arrows indicates the cells that must be in state \(+1\) in order to \textit{activate} the signal.
        \textbf{Top right:} Wire initialized with a TRUE signal.
        \textbf{Bottom left:} Wire initialized with a TRUE signal after one step.
        \textbf{Bottom right:} Wire with a TRUE signal after some steps.} 
        \label{fig:hor_connected_wires}
    \end{figure}

    Next, we show how to simulate logic gates. Since we focus exclusively on monotone circuits, 
    we only need to simulate conjunction and disjunction gates with fan-in \(2\) and fan-out equal to \(2\).
    The conjunction and disjunction gadgets are illustrated in 
    Figure \ref{fig:connected_and}. 
    A conjunction gate outputs TRUE only when both inputs are TRUE, whereas the
    disjunction gate outputs TRUE if at least one input is TRUE.
    Disjunction gates can also act as signal duplicators.

    \begin{figure}[t]
        \centering
        \includegraphics[scale = .3]{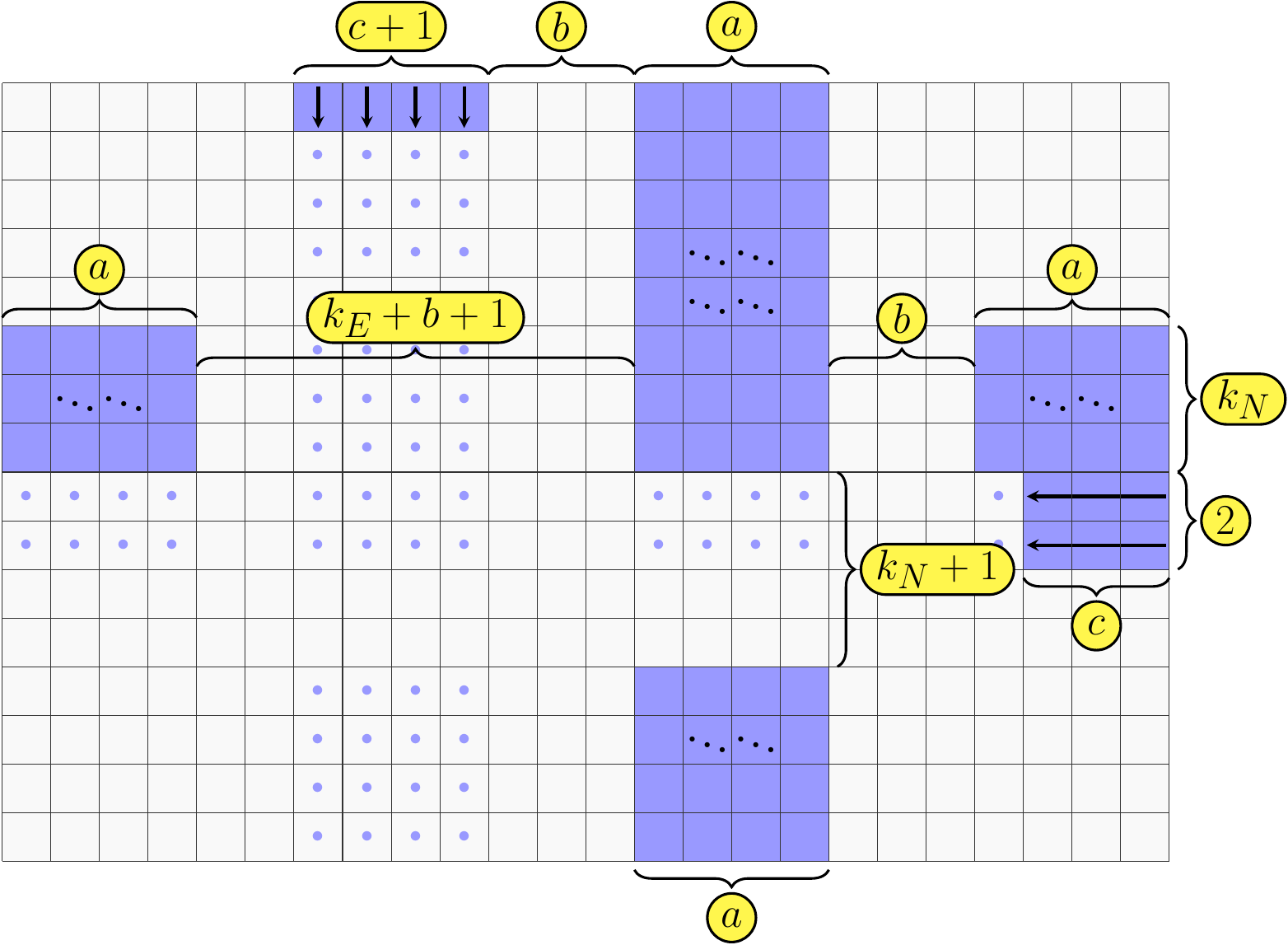}\\[1em]
        \includegraphics[scale = .3]{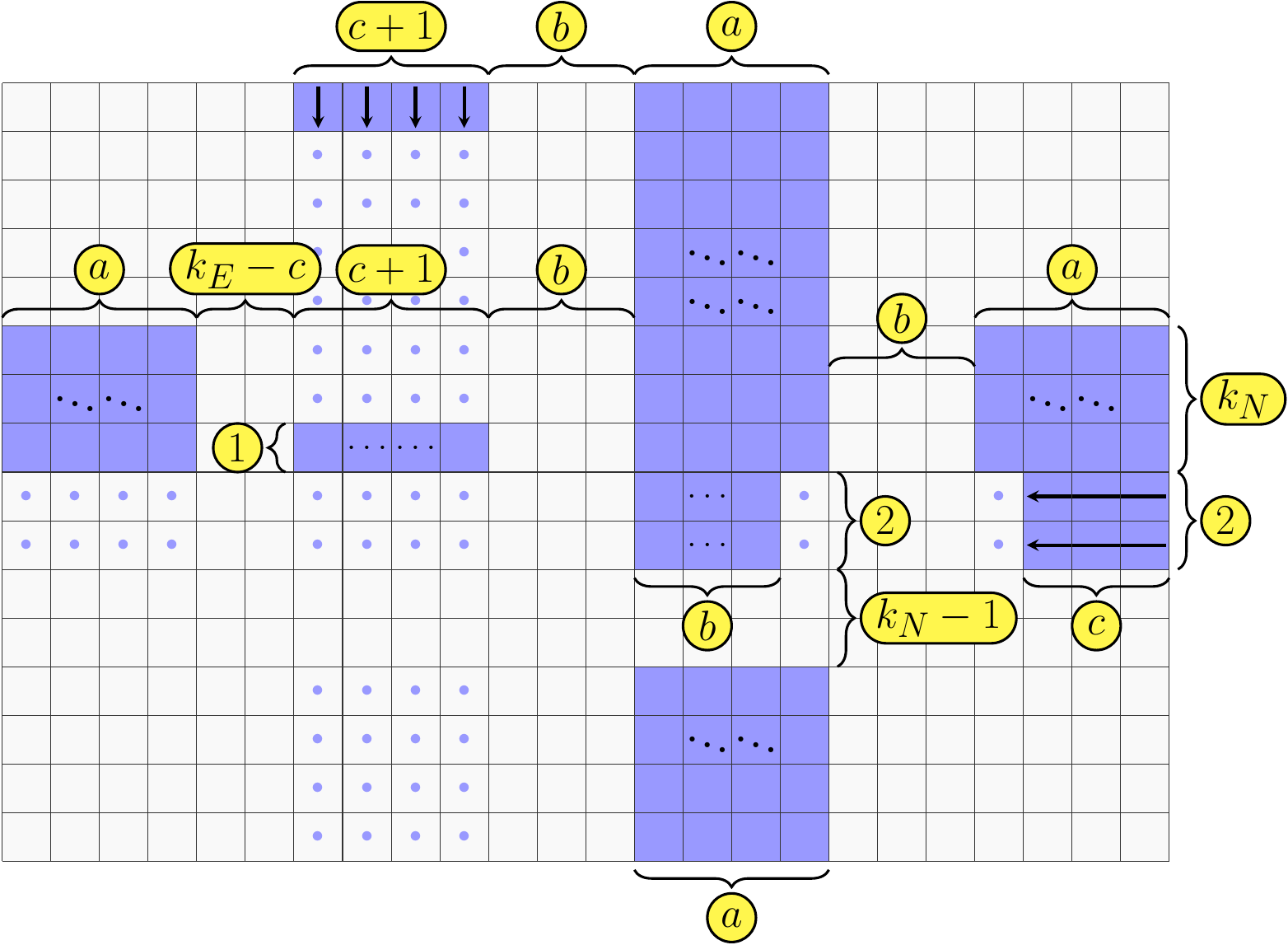}
        \caption{
          Top: AND gadget for contiguous L-shaped neighborhoods.
          Bottom: OR gadget for contiguous L-shaped neighborhoods.
          The cells that eventually go to the freezing state $+1$ are dotted.
        }
        \label{fig:connected_and}
        \label{fig:connected_or}
    \end{figure}

    \begin{figure}[t]
        \centering
        \includegraphics[scale = .3]{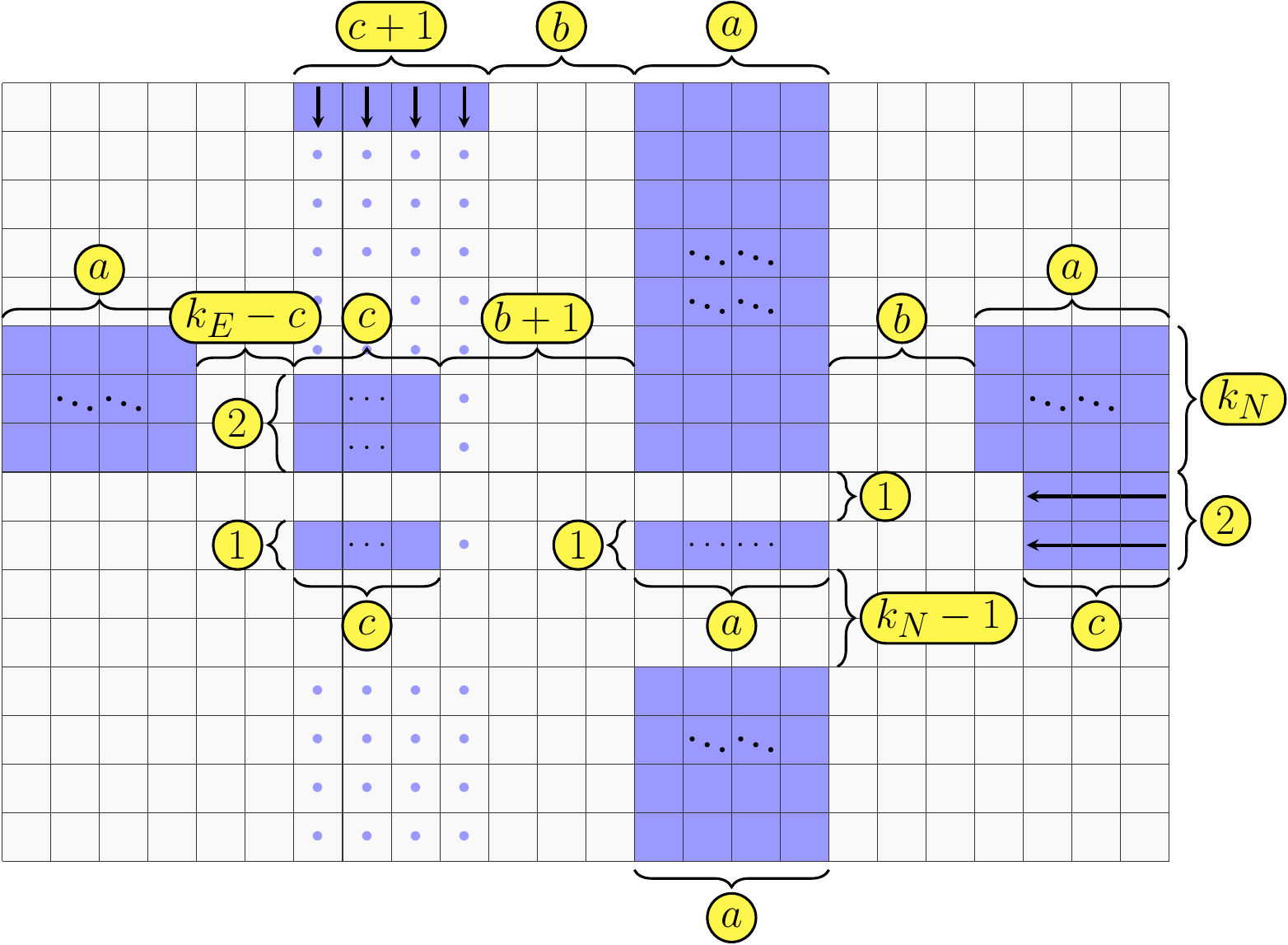}\\[1em]
        \includegraphics[scale = .3]{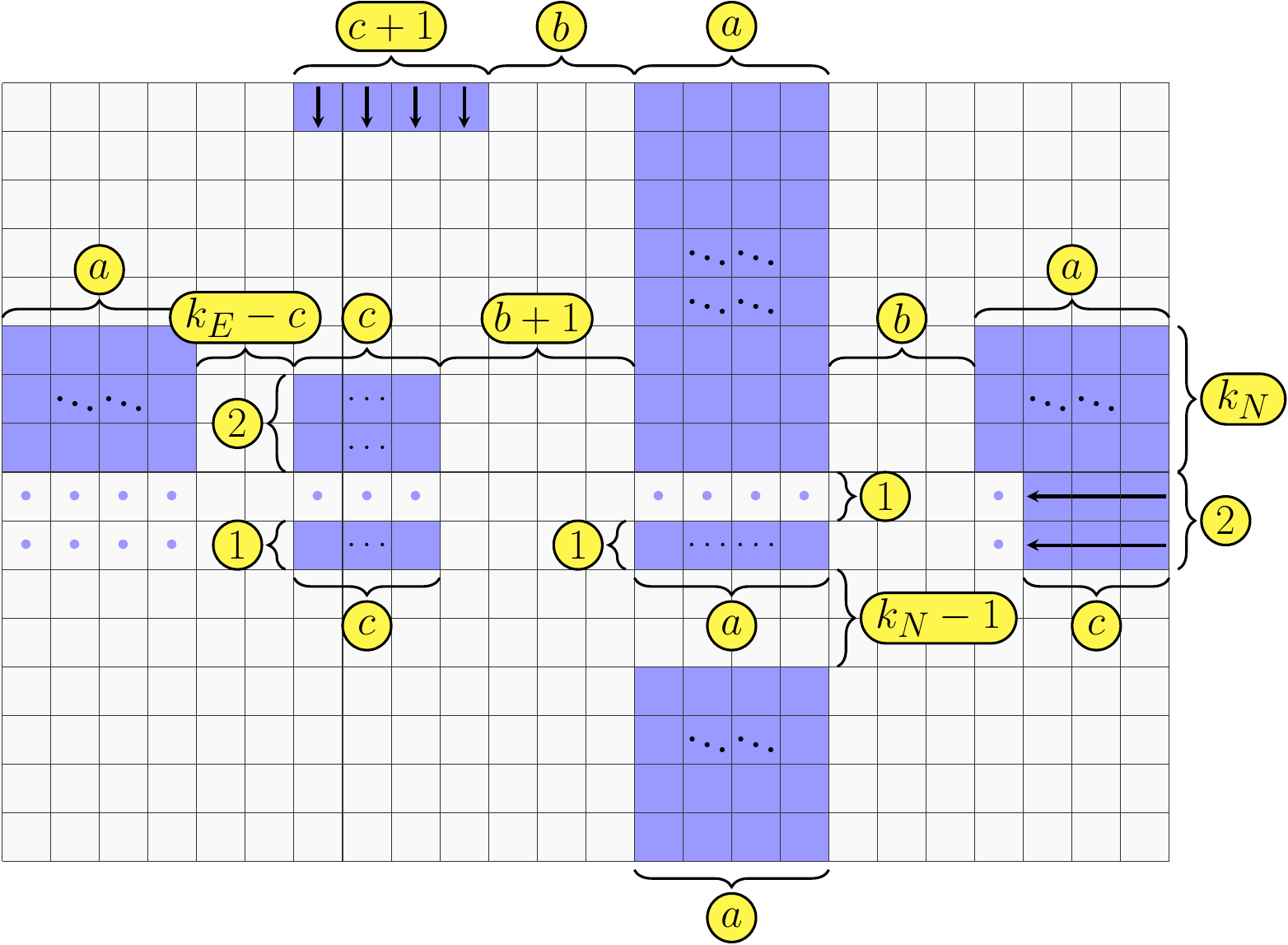}
        \caption{
          Crossover gadget for contiguous L-shaped neighborhoods.
          Top: from north to south.
          Bottom: from west to east.
          The cells that eventually go to the freezing state $+1$ are dotted.
        }
        \label{fig:connected_co}
    \end{figure}

    Finally, to simulate arbitrary circuit, signal crossings are necessary.
    To achieve this, we introduce a \textit{crossover} gadget (Figure \ref{fig:connected_co})
    that allows the east signal to pass exclusively through the west output,
    and the north input signal to pass exclusively through the south output.
    When both signals are TRUE, the crossover gadget behaves as a conjunction gate. 
    Importantly, no prior signal coordination is required.
\end{proof}

Similarly as with Theorem~\ref{thm:toom_nc} and~\ref{thm:size1_nc}, we generalize
Theorem~\ref{thm:connected} to LFMCA with independent subgrids.

\begin{theorem}\label{thm:periodic}
    Let \(S_E\) and \(S_N\) be sets satisfying:
    \begin{itemize}
        \item \(|S_E| \ge 2\),
        \item \(|S_N| \ge 2\),
        \item \((\exists p \in \N): S_E=\{p,2p,\dots,|S_E|p\}\), and 
        \item \((\exists p' \in \N): S_N=\{p',2p',\dots,|S_N|p'\}\), 
    \end{itemize}
    then the \Pred[\(S_N\)-\(S_E\)-] problem is {\PC}.
\end{theorem}

\begin{proof}
    We can apply Lemma~\ref{lemma:neigh_equiv} (with the $p$ and $p'$ from this statement),
    and the construction of Theorem~\ref{thm:connected}
    on any of the subgrids because in it the neighborhood is contiguous
    (provided that the size $n$ is large enough, one can even embed $pp'$
    different circuit simulations in the grid $V_n$,
    but for the purpose of this construction is it sufficient to take $n$ to be a multiple of $p$).
\end{proof}

\subsection{Non-Contiguous L-Shaped Neighborhoods with Set Size \(2\)}

Our focus now shifts to the L-shaped neighborhoods defined by \(S_N\)
and \(S_E\) such that \(|S_N| = |S_E| = 2\). These neighborhoods are
particularly significant due to their minimal size. To the best of
our knowledge, they are the smallest L-shaped neighborhoods for which
the prediction problem is \PC.

\begin{theorem}\label{thm:non_connected}
    Let \(S_E = \{i_E, j_E\}\) and \(S_N = \{i_N, j_N\}\) be sets such that
    the integers \(i_E,j_E,i_N,j_N\) satisfy:


    \begin{itemize}
        \item $j_N \neq 2i_N$,
        \item $j_E \neq 2i_E$,
        \item \(0 < i_E < j_E-1\), and
        \item \(0 < i_N < j_N-1\).
    \end{itemize}
    Then, the \Pred[\(S_N\)-\(S_E\)-] problem is {\PC}.

\end{theorem}

\begin{proof}
    As in Theorem~\ref{thm:connected}, we reduce from the MCVP.
    To accomplish this, we construct conjunction, disjunction and crossover gadgets,
    illustrated respectively in Figures 
    \ref{fig:non_connected_and} and \ref{fig:non_connected_co}. 
    
    Note that to initialize a signal as TRUE, 
    it suffices for one cell to be in state \(+1\).
    It is crucial that this cell is positioned immediately to the north (for signal columns)
    or east (for signal rows);
    otherwise, the gadget will malfunction.
    This requirement enables placing gadgets adjacent to one another
    (side by side, above, or below).

    Finally, note that the size of each gadget depends on a fixed value \(k > 0\),
    which has to be consistent across all gadgets.
\end{proof}

\begin{figure}[t]
    \centering
    \includegraphics[scale = .3]{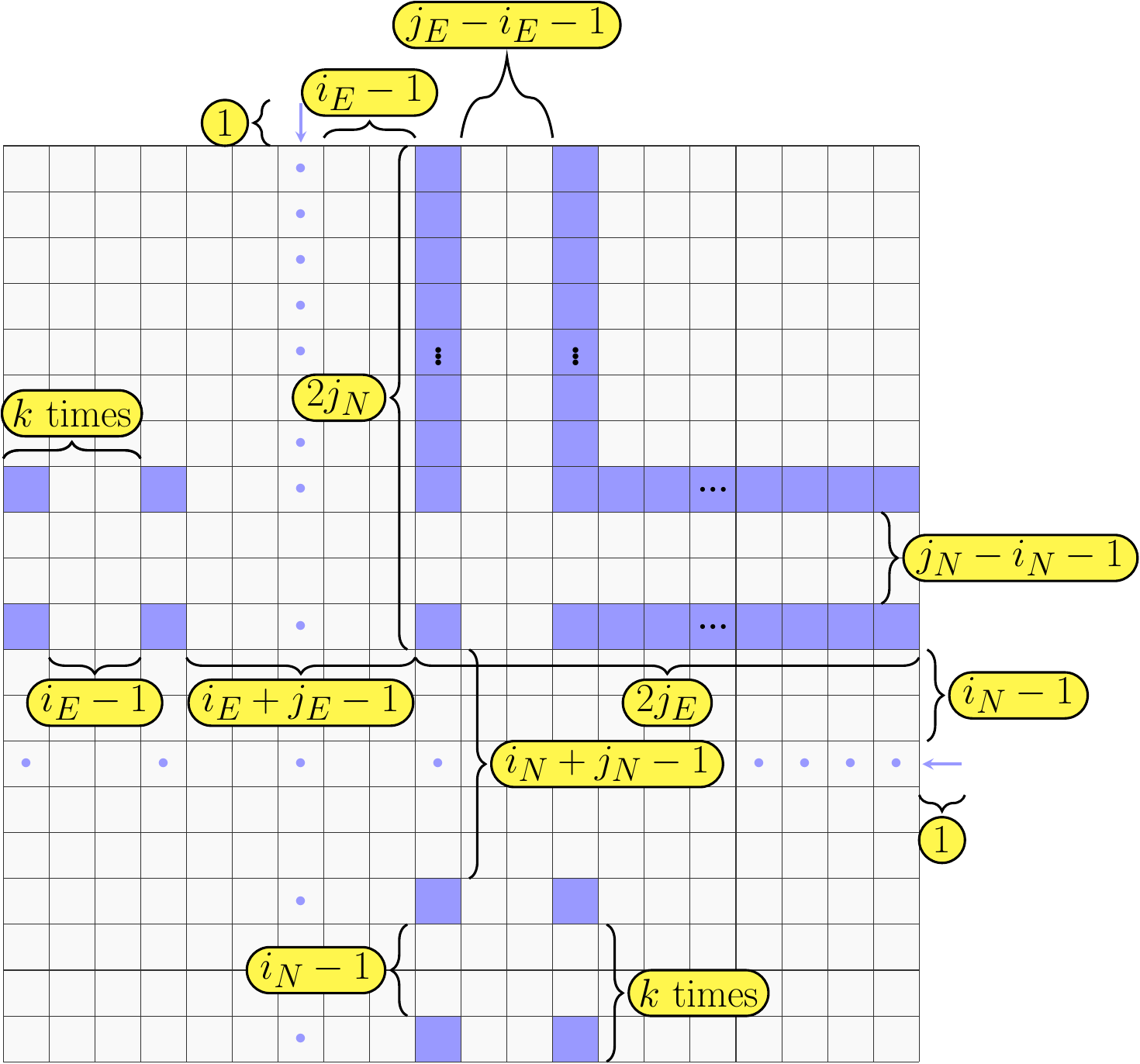}\\[1em]
    \includegraphics[scale = .3]{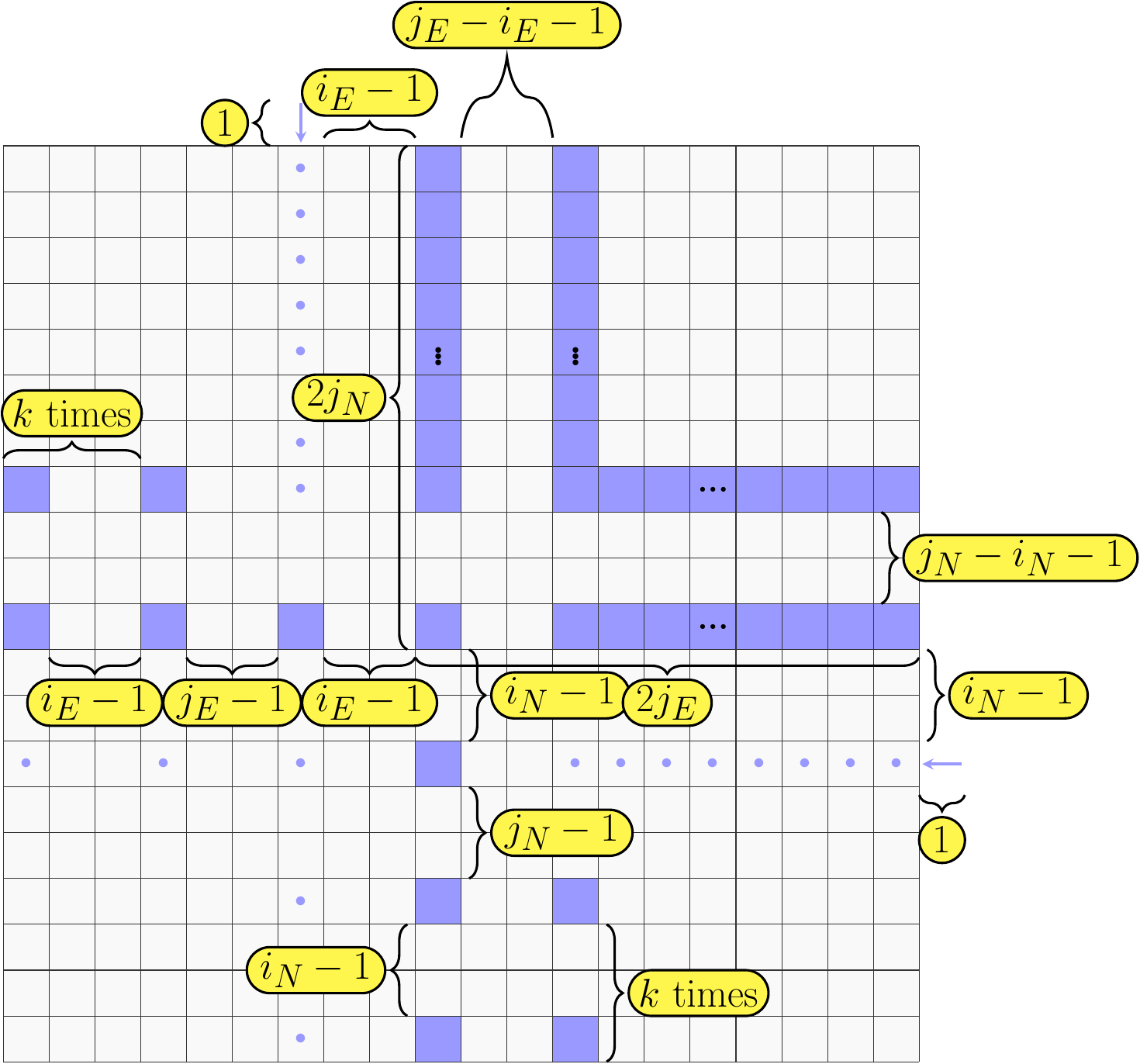}
    \caption{
      Top: AND gadget for non-contiguous L-shaped neighborhoods with \(|S_E| = |S_N| = 2\).
      Bottom: OR gadget for the same neighborhoods.
      The cells that eventually go to the freezing state $+1$ are dotted.
    }
    \label{fig:non_connected_and}
    \label{fig:non_connected_or}
\end{figure}

\begin{figure}[t]
    \centering
    \includegraphics[scale = .3]{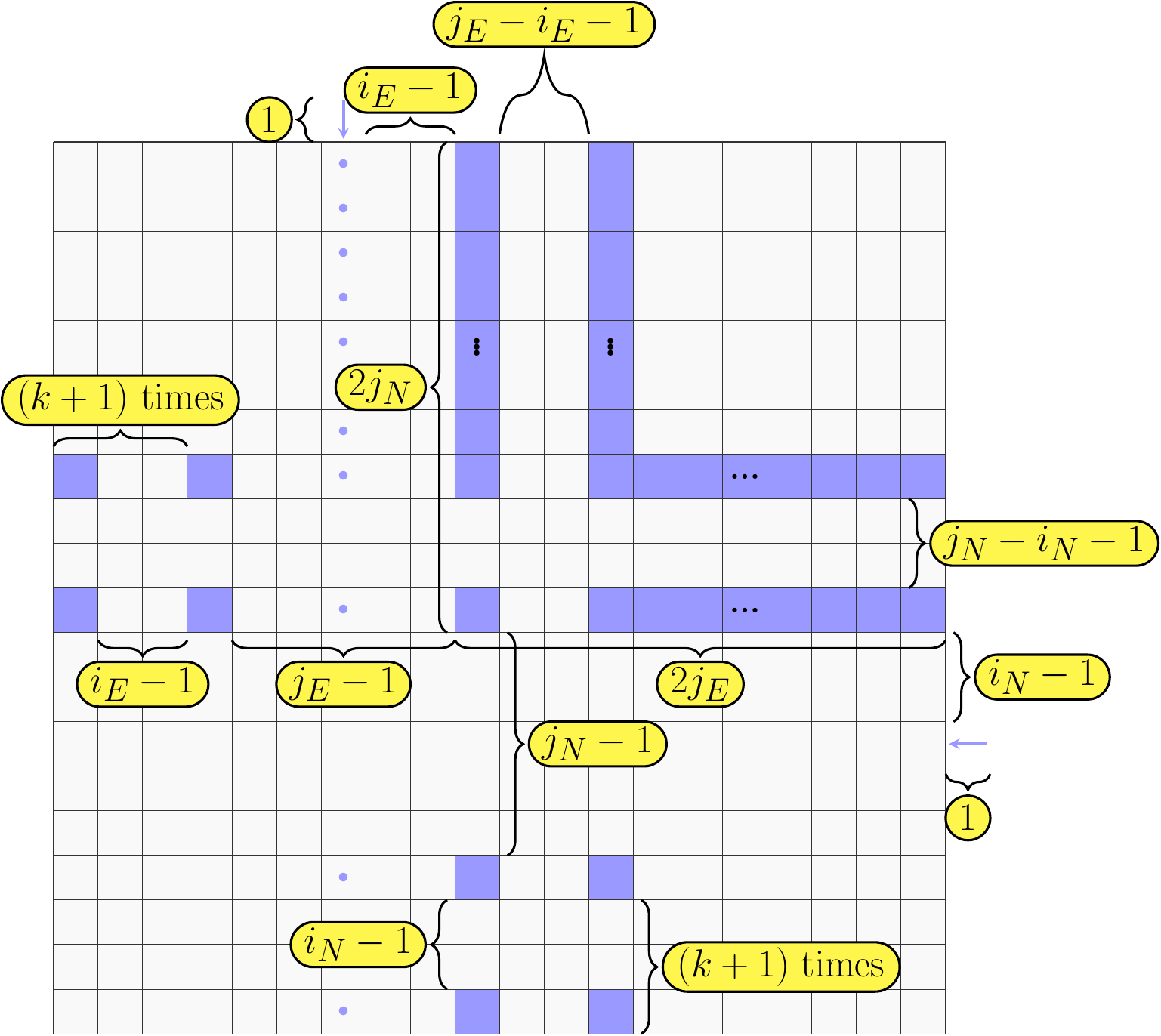}\\[1em]
    \includegraphics[scale = .3]{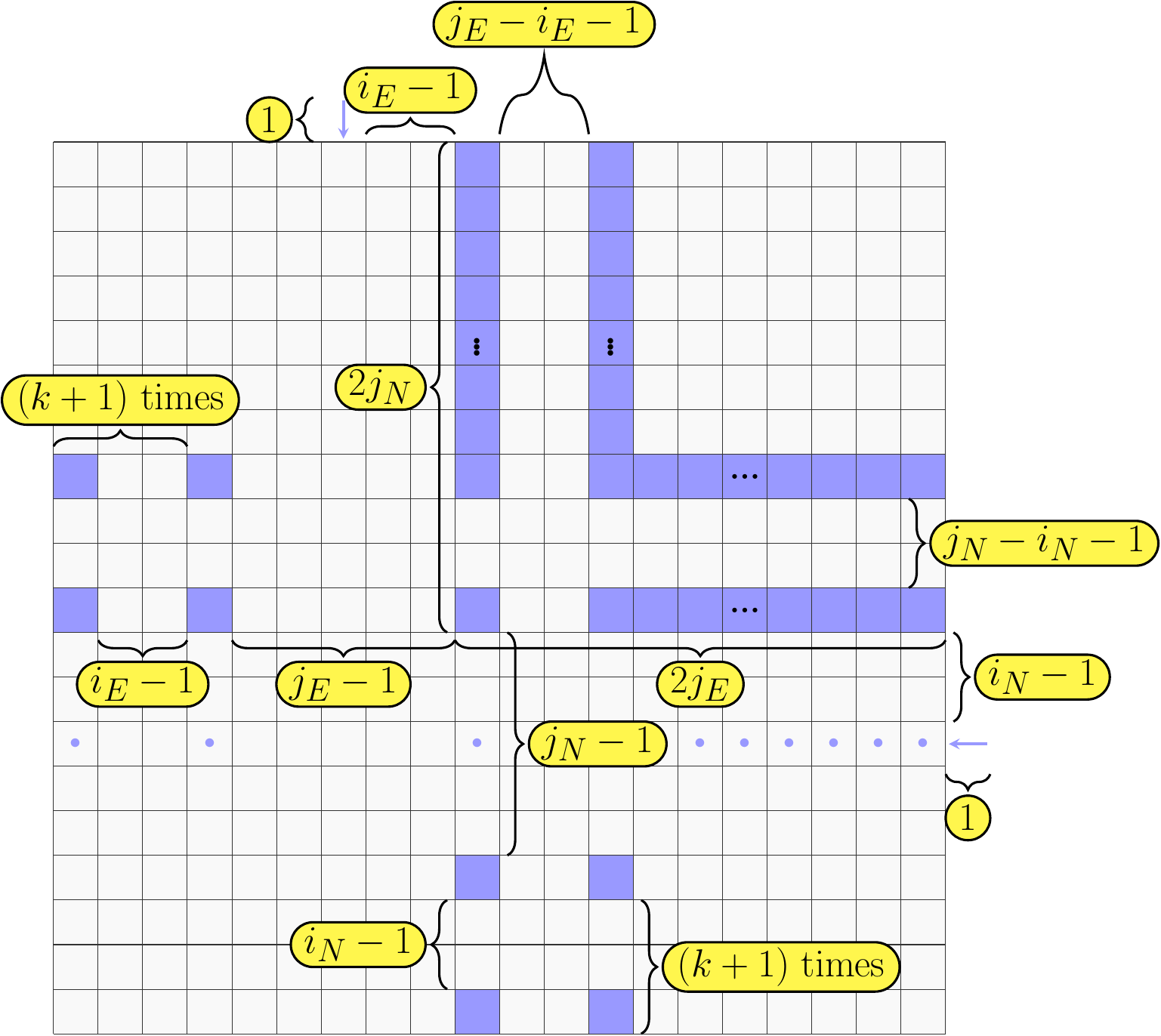}
    \caption{
      Crossover gadget for non-contiguous L-shaped neighborhoods with \(|S_E| = |S_N| = 2\).
      Top: from north to south.
      Bottom: from west to east.
      The cells that eventually go to the freezing state $+1$ are dotted.
    }
    \label{fig:non_connected_co}
\end{figure}

\FloatBarrier
\section{Conclusions}

In this article, we have analyzed the computational complexity of the \Pred
problem in freezing majority cellular automata with L-shaped neighborhoods.
We have introduced a general definition for these neighborhoods by using two
sets of natural numbers $S_E$ and $S_N$, which governs the spatial influence
of each cell along the automaton.

Our results reveal a contrast in the complexity of the \Pred problem depending on the size of the neighborhood.
For the smallest L-shaped neighborhood, we proved that the \Pred problem
belongs to {\NC}, offering an efficient parallel algorithm.
This highlights that under minimal settings, the problem can be efficiently
solved, specially under parallel computing environments.

On the other hand, for (the here-considered) larger L-shaped neighborhoods,
we demonstrated that the \Pred problem becomes considerably harder.
Specifically, we proved that it is {\PC} for three different families of
L-shaped neighborhoods. To achieve this, we construct explicit circuitry
to perform a reduction from the monotone circuit value problem.
These findings show the crucial role that neighborhood size plays on
the complexity of the FMCA.

In conclusion, these results provide a better understanding on how neighborhood
size and structure can affect the overall complexity of the FMCA,
contributing to a broader study of cellular automata and computational complexity.

Note that the \Pred problem remains open for the L-shaped neighborhoods
where \(|S_E| > 1\) and \(|S_N| = 1\). We conjecture that for sufficiently large
\(|S_E|\), the problem belongs to {\NC}. This intuition arises from the
expectation that, as \(S_E\) grows, the CA will resemble a one-dimensional
freezing CA.

Future directions of research could include different neighborhood geometries,
variations of the local rule or considering different updating schedules.
Practical applications of these findings in parallel computation and system
modeling are also open areas for exploration.

\section*{Acknowledgments}

\sloppy
This work received support from the following projects:
ECOS-ANID ECOS240020 (P.M.),
STIC-AMSUD ECODIST AMSUD240005 (P.M.),
FONDECYT 1230599 (P.M.),
FONDECYT 1250984 (E.G.),
Centro de Modelamiento Matemático FB210005 (P.M. and E.G.),
BASAL funds for centers of excellence from ANID-Chile (P.M. and E.G.)
ANID-MILENIO-NCN2024 103 (P.M. and E.G.),
ANR-24-CE48-7504 ALARICE (K.P.),
HORIZON-MSCA-2022-SE-01 101131549 ACANCOS (K.P.),
and STIC AmSud CAMA 22-STIC-02 (Campus France MEAE) (K.P.).

\bibliographystyle{plain}
\bibliography{refs}

\end{document}